\documentclass[draft]{article}

\makeatletter
\renewcommand*\@fnsymbol[1]{\the#1}
\makeatother

\usepackage{amsmath}
\usepackage{amsfonts}
\usepackage{amsthm}
\usepackage[english]{babel}
\usepackage{latexsym}
\usepackage[hmargin=3cm,vmargin=3.5cm]{geometry}
\usepackage{enumerate}
\usepackage{nicefrac}
\usepackage{mathrsfs}
\usepackage[affil-it]{authblk}
\usepackage{bbm}
\usepackage{amssymb}
\usepackage{comment}

\usepackage[affil-it]{authblk}
\theoremstyle{plain}
\newtheorem{theorem}{Theorem}[section]
\newtheorem{lemma}[theorem]{Lemma}
\newtheorem{proposition}[theorem]{Proposition}

\theoremstyle{definition}
\newtheorem{definition}[theorem]{Definition}
\theoremstyle{remark}
\newtheorem{remark}[theorem]{Remark}
\newtheorem{example}[theorem]{Example}

\def\cA{\mathcal{A}}

\def\cF{\mathcal{F}}

\def\cD{\mathcal{D}}
\def\cG{\mathcal{G}}

\def\cC{\mathcal{C}}

\def\cX{\mathcal{X}}

\def\R{\mathbb{R}}

\def\P{\mathbb{P}}
\def\E{\mathbb{E}}
\def\N{\mathbb{N}}

\def\e{\varepsilon}

\def\indic{1}

\DeclareMathOperator*{\esssup}{ess\,sup}

\newcommand{\lnorm}{\left\|}
\newcommand{\rnorm}{\right\|}
\newcommand{\probp}{\P}

\newcommand{\stle}{\ge_{\rm st}}

\newcommand{\closure}{\mathop{\rm cl}\nolimits}
\newcommand{\VaR}{\mathop {\rm VaR}\nolimits}

\newcommand{\ES}{\mathop {\rm ES}\nolimits}

\newcommand{\rec}{\mathop {\rm rec}\nolimits}
\newcommand{\Rec}{\mathop {\rm rec}\nolimits}

\newcommand*{\longhookrightarrow}{\ensuremath{\lhook\joinrel\relbar\joinrel\rightarrow}}

\title{Diversification, protection of liability holders and\\
regulatory arbitrage}

\author{\sc{Pablo Koch-Medina}\thanks{Email: \texttt{pablo.koch@bf.uzh.ch}}\,, \sc{Cosimo Munari}\,\thanks{Email: \texttt{cosimo.munari@bf.uzh.ch}}\,, \sc{Mario \v{S}iki\'c}\,\thanks{Email: \texttt{mario.sikic@bf.uzh.ch}}}
\affil{Center for Finance and Insurance, University of Zurich, Switzerland}

\begin{document}
\maketitle

\parindent 0em \noindent

\begin{abstract}
\noindent Any solvency regime for financial institutions should be aligned with the fundamental objectives of regulation: protecting liability holders and securing the stability of the financial system. The first objective leads to consider surplus-invariant capital adequacy tests, i.e. tests that do not depend on the surplus of a financial institution. We provide a complete characterization of closed, convex, surplus-invariant capital adequacy tests that highlights an inherent tension between surplus-invariance and the desire to give credit for diversification. The second objective leads to requiring consistency of capital adequacy tests across jurisdictions. Of particular importance in this respect are capital adequacy tests that remain invariant under a change of num\'{e}raire. We establish an intimate link between surplus- and num\'{e}raire invariant tests.
\end{abstract}

%% %% %% %% %% %% %% %% %% %% %% %% %% %% %%

\section{Introduction}

One of the major advances in the regulation of financial institutions, be it banks or insurance companies, has been the introduction of risk-sensitive solvency regimes. Examples are the Basel Accord in the banking sector and Solvency II and the Swiss Solvency Test in the insurance sector. Let us recall the typical mathematical framework of risk-sensitive solvency regimes which we will call the ``classical model'' in this paper; see for instance Artzner, Delbaen, Eber \& Heath \cite{ADEH1999}, F\"{o}llmer \& Schied \cite{FoellmerSchied2011} or, for an account in the spirit of this paper, Farkas, Koch-Medina \& Munari \cite{FarkasKochMunari2014}. {\em Capital positions} -- assets net of liabilities -- of financial institutions are assumed to belong to a space $\cX$ of random variables on some probability space $(\Omega,\cF,\probp)$ representing the {\em future states of the economy}. At any state $\omega\in\Omega$, an institution with capital position $X\in\cX$ will be able to meet its obligations whenever $X(\omega)\ge0$ and will default whenever $X(\omega)<0$. A financial institution is deemed to be adequately capitalized if its capital position belongs to a pre-specified subset $\cA$ of $\cX$, called the {\em acceptance set} or the {\em capital adequacy test}. Finally, {\em risk measures} describe the minimum cost of meeting the capital adequacy test by raising capital and investing it in a reference instrument, often assumed to be cash.

\subsubsection*{Microprudential regulation and surplus invariance}

Capital adequacy tests are primarily an instrument of microprudential regulation, i.e. their main purpose is to help protect liability holders. From this perspective it would seem natural for capital adequacy tests to be {\em surplus invariant}, i.e. for a financial institution with capital position $X$ the size of the {\em surplus}
\[
X^+ := \max\{X,0\}\,,
\]
which benefits only the institution's owners, should have no impact on whether the institution passes or fails the test. In other words, acceptability should only depend on the default option
\[
X^- := \max\{-X,0\}
\]
which, in case of a company with limited liability, represents the difference between the contractual and the actual liability payment. Formally, this entails requiring that
\[
X\in\cA, \ \ Y\in\cX, \ \ Y^-\leq X^- \ \implies \ Y\in\cA\,.
\]

\smallskip

While surplus invariance is a reasonable requirement, it is also reasonable --- and in the interest of liability holders --- that capital requirements give credit for diversification. This is because an institution that diversifies its risk exposures can reduce the costs related to holding capital, costs that are ultimately borne by liability holders. It is well-known that, for an acceptance set, the financial requirement of giving credit for diversification is captured by the mathematical property of {\em convexity}; see for instance F\"{o}llmer \& Schied \cite{FoellmerSchied2002} or Frittelli \& Rosazza-Gianin \cite{FRG2002}. In this paper we provide a characterization of convex, surplus-invariant acceptance sets which, as explained further down, will highlight an interesting and important tension that exists between the requirement that capital adequacy tests be surplus invariant and that, at the same time, they give credit for diversification.

\subsubsection*{Macroprudential regulation and num\'{e}raire invariance}

While capital adequacy tests are part of the toolkit of microprudential regulation, they should ideally also support or, at the very least, not undermine macroprudential regulation, whose objective is to secure the stability of the financial system. By its very nature, the effectiveness of macroprudential regulation depends on how well harmonized regulatory frameworks across boundaries are. In an ideal world, no regulatory capital arbitrage opportunities would exist, i.e. it should no be possible for a financial institution to pass from being unacceptable to being acceptable by merely changing to a different jurisdiction. This implies that capital adequacy tests across boundaries should just be ``translations'' of each other by the appropriate exchange rate: Given capital adequacy tests $\cA_1$ and $\cA_2$ in two different currencies, one should have $\cA_2=R\cA_1$, where $R$ is the (stochastic) exchange rate from currency $1$ into currency $2$. However, in today's global regulatory system this translation step is not done: Each jurisdiction uses the ``same'' capital adequacy test --- based on Value-at-Risk ($\VaR$) or Expected Shortfall ($\ES$) --- but applies it in its own currency. As explained in Koch-Medina \& Munari~\cite{KM2016}, $\VaR$-based tests (leaving aside all of their potential deficiencies) would lead to a consistent global regulatory regime but $\ES$ based tests (leaving aside all of their potential merits) would not. While $\ES$-based tests are convex (indeed coherent), $\VaR$-based tests fail to be convex. Therefore it is interesting to ask whether there exist any convex capital adequacy tests that could be used in any currency without requiring translation. These capital adequacy tests will be called {\em num\'{e}raire invariant} and their formal definition is as follows: For every change of currency (or, more generally, num\'{e}raire), represented by a bounded random variable $R$ that is strictly positive almost surely, we require that
\[
X\in\cA \ \implies \ RX\in\cA\,.
\]
As it transpires, num\'{e}raire invariance and surplus invariance are intimately related. In fact, convex, num\'{e}raire-invariant capital adequacy tests coincide with coherent, surplus-invariant tests.

\subsubsection*{Structure of the paper and main results}

The setting of our paper is sufficiently general to cover the typical choices for the ambient space $\cX$ encountered in the literature. Indeed, we will work in an ambient space $\cX$ belonging to a class of spaces of random variables that includes all Orlicz hearts (and, hence, all $L^p$ spaces for $1\leq p<\infty$), $L^\infty$ when equipped with the weak-star topology, but also the non-locally convex spaces $L^p$ for $0\leq p<1$.

\smallskip

Section~3 contains our main results. After introducing and proving some basic properties of surplus-invariant acceptance sets, we turn to their characterization in case they are closed and convex. Theorem~\ref{theorem: dual repr convex xs acc} provides a dual representation that relies on the fact that any closed, surplus-invariant acceptance set $\cA\subset \cX$ can be recovered by taking the closure of its intersection with $L^\infty$. This allows us to obtain a ``dual'' representation even if the ambient space is not locally convex, e.g. an $L^p$ space for $0\le p<1$. Armed with this representation we first provide in Lemma~\ref{thm: coherent structures} a generalization to our setting of a result by Koch-Medina, Moreno-Bromberg \& Munari \cite{KMM2015} characterizing weak-star closed, surplus-invariant acceptance sets on $L^\infty$ that are coherent, i.e. convex cones. The result states that a closed, coherent acceptance set is surplus invariant if and only if we find a measurable set $A$ with strictly positive probability such that
\begin{equation}
\label{equivalence coherent intro}
X\in\cA \ \iff \ 1_AX\geq0
\end{equation}
holds for every $X\in\cX$. In other words, closed, coherent, capital adequacy tests are surplus-invariant if and only if we impose no restrictions on $A^c$ and disallow any defaults in the stress scenarios, i.e.  in the states belonging to $A$. The problem with these capital adequacy tests is that either $\probp(A)<1$ and they ignore what happens outside $A$, or $\probp(A)=1$ and the tests disallow defaults in all states of the world.

\smallskip

The characterization in \eqref{equivalence coherent intro} is used to show in Theorem \ref{main theorem} that if we drop coherence but retain closedeness and convexity, $\cA$ is surplus invariant if and only if we find a partition $\{A,B,C\}$ of $\Omega$ consisting of measurable sets such that
\begin{equation}
\label{equivalence intro}
X\in\cA \ \iff \ 1_AX\geq0 \ \ \,\mbox{and} \ \ -1_BX^-\in\cD_B
\end{equation}
for every $X\in\cX$, where $\cD_B$ is a closed, convex set that is tight, i.e. bounded in probability. The decomposition of $\Omega$ into the three classes of scenarios implied by the partition $\{A,B,C\}$ has a clear financial interpretation: A financial institution is adequately capitalized if and only if
\begin{itemize}
  \item no default is allowed on $A$,
  \item a ``controlled'' form of default is allowed on $B$,
  \item no constraint is required on $C$.
\end{itemize}
Incidentally, Example~3.3 shows that the above characterization does not hold if we equip $L^\infty$ with the strong topology.

\smallskip

To obtain a better sense of just how controlled defaults must be on the scenarios belonging to $B$, we can exploit a characterization of tightness, given in Proposition \ref{prop: tight and stoch bounded}, in terms of stochastic boundedness with respect to first order stochastic dominance. As a result we can show that there exists a random variable $X^\ast\in L^0$ such that
\[
\VaR_\alpha(1_BX) \leq \VaR_\alpha(X^\ast) \ \ \mbox{for all $X\in\cA$}, \ \alpha\in(0,1)\,,
\]
where the {\em Value-at-Risk} of any $X\in L^0$ is defined as usual by
\[
\VaR_\alpha(X) := \inf\{t\in\R \,; \ \probp(X+t<0)\leq\alpha\}\,.
\]
It follows that, on $B$, the $\VaR$ of every acceptable position is rigidly controlled by the $\VaR$ of a single bounding random variable $X^\ast$ for every level $\alpha\in(0,1)$.

\smallskip

In Section~4 we discuss num\'{e}raire invariance and show that it is intimately related to surplus invariance. Indeed, in Proposition \ref{proposition: numinv is equiv to supinv} we show that a closed capital adequacy test is num\'{e}raire invariant if and only if it is a conical, surplus-invariant acceptance set. Hence, it follows that, under closedness, if we simultaneously require convexity and num\'{e}raire invariance we end up allowing no defaults on a measurable set $A$ and imposing no constraints on the complement $A^c$.

\smallskip

The limited choice of convex capital adequacy tests that satisfy either surplus or num\'{e}raire invariance reflects the tension we mentioned at the beginning of the introduction: Convexity reduces num\'eraire invariant capital adequacy tests to simple stress tests. Thus, if we want to avoid having states of the world in which financial institutions are completely unconstrained, we are left with a single possible test that is passed only by those institutions that will never default. This tension lessens considerably if we drop num\'{e}raire invariance while retaining surplus-invariance: Convexity now allows for tests where controlled defaults are possible and there are no uncontrolled scenarios. In Example \ref{cor: construction of surplus invariance via ES} we provide a constructive way to define convex, surplus-invariant capital adequacy tests based on Expected Shortfall that have this property.

\smallskip

In a final section, we qualify the implications of our results by suggesting that they may point to a deficiency of the classical model. Indeed, due to its rudimentary nature, the classical model does not allow for the possibility that an institution running a surplus may provide a benefit to its liability holders or society at large. Enhancing the classical model to accommodate such features may constitute a potential area for future research.

\smallskip

Our work generalizes results obtained recently in Koch-Medina, Moreno-Bromberg \& Munari \cite{KMM2015} for coherent acceptance sets in the context of spaces of bounded positions in two ways. First, this paper goes beyond coherence and allows for convex acceptance sets and, second, we allow a more general class of ambient spaces. As discussed in detail in that paper, the notion of surplus-invariance is related to similar notions discussed by Staum \cite{Staum2013} and by Cont, Deguest \& He \cite{ContDeguestHe2013}.

%%%%%%%%%%%%%%%%%%%%%%%%%%%%%%%%%%%%%%%%%%%%%%%%%%%

\section{Financial positions and acceptance sets}

Throughout this paper, $(\Omega,\cF,\probp)$ will denote a fixed probability space. As usual, we will identify random variables that coincide almost surely (with respect to $\probp$). The indicator function of a set $A\in\cF$ is denoted by $1_A$. The space $L^0:= L^0(\Omega,\cF,\probp)$ is endowed with the natural almost sure pointwise ordering, i.e. for $X,Y\in L^0$ we write $Y\ge X$ whenever $\probp(Y\geq X)=1$. A random variable $X\in L^0$ is said to be {\em positive} if $X\ge 0$ and {\em negative} if $X\leq0$. The {\em positive cone} and the {\em negative cone} are the closed, convex cones defined by
\[
L^0_+:=\{X\in L^0 \,; \ X\ge 0\} \ \ \ \mbox{and} \ \ \ L^0_-:=\{X\in L^0 \,; \ X\le 0\}\,.
\]
When equipped with the usual topology of convergence in probability, $L^0$ becomes a Frech\'{e}t lattice, i.e. a completely metrizable locally solid Riesz space. For $\cA\subset L^0$ we define $\cA_+:=\cA\cap  L^0_+$ and $\cA_-:=\cA\cap L^0_-$. Moreover, for any $X\in L^0$ we set
\[
X^+:=\max\{X,0\} \ \ \ \mbox{and} \ \ \ X^-:=\max\{-X,0\}\,.
\]
In particular, $X=X^+-X^-$.

\smallskip

Our ambient space $\cX$ will be either $L^\infty:= L^\infty(\Omega,\cF,\probp)$ equipped with the weak-star topology $\sigma(L^\infty,L^1)$ or any linear subspace of $L^0$ that is a Frech\'{e}t lattice satisfying the {\em $\sigma$-Lebesgue property} (see \cite{AliprantisBurkinshaw2003})
\begin{equation}
\label{sigma lebesgue}
X_n\downarrow 0 \ \mbox{a.s.} \ \implies \ X_n\to 0 \ \mbox{in} \ \cX,
\end{equation}
and the following dense, continuous embeddings
\begin{equation}
\label{embeddings}
L^\infty(\Omega) \overset{d}{\longhookrightarrow} \cX\overset{d}{\longhookrightarrow} L^0(\Omega)\,,
\end{equation}
where the continuity of the first embedding is with respect to the sup-norm and not the weak-star topology. Finally, we also require the following {\em localization property}
\begin{equation}
\label{localization}
1_A\cX := \{1_AX \,; \ X\in\cX\} \subset \cX\,.
\end{equation}

\smallskip

\begin{remark}
\label{rem: examples underlying space}
\begin{enumerate}[(i)]
\item The $\sigma$-Lebesgue property implies that every sequence $(X_n)$ in $\cX$ that is monotone and converges a.s. to $X\in\cX$ satisfies $X_n\to X$ in $\cX$.

\item In the context of Banach lattices, the Lebesgue property is nothing but the order continuity of the norm. Hence, the admissible class of ambient spaces includes all Orlicz hearts and, hence, all $L^p$ spaces for $1\leq p<\infty$. However, it does not include $L^\infty$ when equipped with the norm topology. Examples outside the Banach lattice universe include all $L^p$-spaces for $0\le p<1$ (which are not locally convex).

\item Note that, when equipped with the weak-star topology, $L^\infty$ also satisfies the $\sigma$-Lebesgue property \eqref{sigma lebesgue} and the embedding condition \eqref{embeddings}. However, $L^\infty$ with the weak-star topology is not a Frech\'{e}t lattice, since it is not metrizable. Nonetheless, for convex sets, closedeness can still be characterized by sequencial closedness as a well-known result by Grothendieck (Exercise 1 on page 240 in \cite{Grothendieck1973}) states: a convex subset of $L^\infty$ is weak-star closed if and only if it is closed under the operation of taking a.s. limits of norm bounded subsequences.
\end{enumerate}
\end{remark}

\smallskip

We write $\closure_\cX(\cA)$ for the closure of a set $\cA$ in $\cX$. In particular, for a subset $\cA$ of $L^\infty$, $\closure_{L^\infty}(\cA)$ always denotes its weak-star closure.

\smallskip

We consider a one-period economy with dates $t=0$ and $t=T$. At time $t=0$, financial institutions issue liabilities and invest in assets. At time $t=T$ they receive the payoff of the assets and redeem their liabilities. Assets and liabilities are assumed to be denominated with respect to a fixed unit of account, e.g. a fixed currency, and to belong to $\cX$. If $A\in \cX$ and $L\in \cX$ are positive random variables representing the terminal payoff of the institution's assets and liabilities, respectively, we will refer to the random variable $X:=A-L\in \cX$ as the {\em capital position} of the financial institution. We will always assume the owners of the institution have {\em limited liability}, i.e. the institution will default at time $t=T$ whenever the payoff of the assets does not suffice to repay liabilities. A concern of regulators is the risk of financial institutions defaulting on their obligations and one of the key instruments they have to mitigate this risk is to require that financial institutions be {\em adequately capitalized}. Acceptance sets are used to formalize the process of testing for capital adequacy.

\smallskip

Recall that a non-empty, strict subset $\cA\subset \cX$ is called an \emph{acceptance set} or a {\em capital adequacy test} if it is {\em monotone}, i.e.
\[
X\in\cA, \ \ Y\geq X \ \implies \ Y\in\cA.
\]
Acceptance sets that are convex or coherent, i.e. convex cones, are of particular importance because they capture diversification.

%%%%%%%%%%%%%%%%%%%%%%%%%%%%%%%%%%%%%%%%%%%%%%%%%%

\section{Surplus invariance}

Consider a financial institution with capital position $X\in \cX$. The positive random variable $X^+$ is called the {\em (owners') surplus} and the, also positive, random variable $X^-$ is called the {\em (owners') option to default}. In case of a financial institution with limited liability, the surplus represents what belongs to the owners after liabilities have been settled and the option to default the amount by which the institution defaults. More precisely, $X^-$ represents the difference between the contractual and the actual payment to liability holders. Since a capital adequacy test is designed to protect the interests of liability holders, it makes sense that an institution should not pass the test if its default profile is riskier than the default profile of an institution that has been deemed inadequately capitalized. Equivalently, if an institution has been deemed adequately capitalized, any institution with a less risky default profile should also pass the test. This leads to the notion of a ``surplus invariant'' acceptance set.

\begin{definition}
Let $\cA\subset \cX$ be an acceptance set. We say that $\cA$ is {\em surplus invariant} if acceptability does not depend on the surplus of a capital position, i.e.
\[
X\in\cA, \ \ Y^-\leq X^- \ \implies \ Y\in\cA\,.
\]
\end{definition}

\smallskip

We start by providing a list of useful alternative characterizations of surplus invariance, which will be used without explicit reference in the sequel.

\begin{proposition}
\label{proposition: alternative definitions}
Let $\cA\subset\cX$ be an acceptance set. The following statements are equivalent:
\begin{enumerate}[(a)]
\item $\cA$ is surplus invariant;
\item $X\in\cA$ and $Y^-=X^-$ imply $Y\in\cA$;
\item $X\in\cA$ implies $-X^-\in\cA$;
\item $X\in\cA$ and $A\in\cF$ imply $1_AX\in\cA$.
\end{enumerate}
\end{proposition}
\begin{proof}
It is clear that {\em (a)} implies {\em (b)}, which in turn implies {\em (c)}. Now, assume {\em (c)} holds and take $X\in\cA$ and $A\in\cF$. Since $1_AX\geq -X^-$, we immediately conclude that {\em (d)} is satisfied. Finally, assume {\em (d)} holds and take $X\in\cA$. If $Y^-\leq X^-$, then $Y\geq 1_{\{X<0\}}X$ implying that $Y\in\cA$ and showing that $\cA$ is surplus invariant.
\end{proof}

\smallskip

\begin{remark}
The characterization under {\em (b)} establishes the equivalence of the present definition of surplus invariance with the one introduced in Koch-Medina, Moreno-Bromberg \& Munari \cite{KMM2015}.
\end{remark}

\smallskip

The next result shows that a surplus-invariant acceptance set is fully determined by its negative part. A set $\cD\subset\cX_-$ is said to be {\em solid (in $\cX_-$)} if it satisfies
\[
X\in\cD, \ \ Y\in\cX, \ \ X\leq Y\leq0 \ \implies \ Y\in\cD\,.
\]

\begin{proposition}
\label{proposition: first characterization surpinv}
An acceptance set $\cA\subset\cX$ is surplus invariant if and only if
\begin{equation}
\label{first charac}
\cA =\cD+\cX_+
\end{equation}
for some solid set $\cD\subset\cX_-$. In this case, $\cA_-=\cD$. The set $\cA$ is convex if and only if $\cD$ is convex. Moreover, $\cA$ is closed if and only if $\cD$ is closed.
\end{proposition}
\begin{proof}
It is clear that $\cA$ is surplus invariant if it is given as in \eqref{first charac} for some solid set $\cD\subset\cX_-$. On the other hand, if $\cA$ is surplus invariant, then $\cA_-$ is solid by monotonicity. Since $\cA_-+\cX_+\subset\cA$ clearly holds, we only need to show the converse inclusion. Take $X\in\cA$ and note that $-X^-\in\cA$. Since $X=X^+-X^-$, the claim follows. The equivalence between the convexity of $\cA$ and that of $\cD$ is obvious. Finally, it is clear that if $\cA$ is closed so is $\cD$. Conversely, by using the continuity of the lattice operations, one can easily prove that $\cA$ is closed whenever $\cD$ is closed.
\end{proof}

\smallskip

The property of surplus invariance is shared by a variety of important acceptance sets.

\begin{example}[Value-at-Risk]
The {\em Value-at-Risk} of $X\in\cX$ at level $\alpha\in(0,1)$ is defined by
\[
\VaR_\alpha(X) := \inf\{t\in\R \,; \ \probp(X+t<0)\leq\alpha\}\,.
\]
The associated acceptance set
\[
\cA = \{X\in\cX \,; \ \VaR_\alpha(X)\leq0\}
\]
is easily seen to be surplus invariant since
\[
\cA = \{X\in\cX \,; \ \probp(X<0)\leq\alpha\} = \{X\in\cX \,; \ \probp(X^->0)\leq\alpha\}\,.
\]
In particular, acceptability boils down to checking whether the probability of default of $X$ does not exceed the threshold $\alpha$.
\end{example}

\smallskip

\begin{example}[Shortfall Risk]
Consider a nonconstant, increasing, convex function $\ell:\R_+\to\R_+$ satisfying $\ell(0)=0$ and take $c>0$. The acceptance set
\[
\cA = \{X\in\cX \,; \ \E[\ell(X^-)]\leq c\}
\]
is surplus-invariant. For more details we refer to F\"{o}llmer \& Schied \cite{FoellmerSchied2002}, where $\cA$ was called the acceptance set based on {\em shortfall risk} and $\ell$ was referred to as a {\em loss function}.
\end{example}

\smallskip

\begin{example}[Test Scenarios]
Consider a set $E\in\cF$ satisfying $\probp(E)>0$. The acceptance set
\[
\cA = \{X\in\cX \,; \ 1_EX\geq0\}
\]
is immediately seen to be surplus-invariant. In this case, acceptability is equivalent to requiring no default on the event $E$. In particular, if $\probp(E)=1$, then $\cA=\cX_+$ so that acceptability is equivalent to requiring no default in almost every state of the world.
\end{example}

\smallskip

\begin{example}[Expected Tail Loss]
The {\em Expected Shortfall} of $X\in\cX$ at level $\alpha\in(0,1)$ is defined by
\[
\ES_\alpha(X) := \frac{1}{\alpha}\int_0^\alpha\VaR_\beta(X)d\beta\,.
\]
The {\em Expected Tail Loss} of $X$ at level $\alpha$ is defined as $\ES_\alpha(-X^-)$ and was studied in Cont, Deguest \& He \cite{ContDeguestHe2013} in the context of loss-based risk measures. The acceptance set based on Expected Tail Loss at level $c\geq0$ is therefore given by
\[
\cA = \{X\in\cX \,; \ \ES_\alpha(-X^-)\leq c\}\,.
\]
It is immediate to see that $\cA$ is surplus invariant. As already remarked in Koch-Medina, Moreno-Bromberg \& Munari \cite{KMM2015}, the choice $c=0$ leads to the trivial case $\cA=\cX_+$. Indeed, the sensitivity of Expected Shortfall implies that $\ES_\alpha(-X^-)>0$ as soon as $X$ is nonzero.
\end{example}

\smallskip

The most prominent example of an acceptance set that fails to be surplus invariant is the one based on Expected Shortfall.

\begin{example}[Expected Shortfall]
The acceptance set based on Expected Shortfall at level $\alpha\in(0,1)$
\[
\cA = \{X\in\cX \,; \ \ES_\alpha(X)\leq0\}
\]
is not surplus invariant. Indeed, it is not difficult to verify that
\[
\ES_\alpha(X) = \frac{1}{\alpha}\int_0^{\probp(X<0)}\VaR_\beta(-X^-)d\beta+\frac{1}{\alpha}\int_{\probp(X<0)}^\alpha\VaR_\beta(X^+)d\beta
\]
holds for every $X\in\cX$, showing that $\ES_\alpha(X)$ does depend, whenever $\probp(X<0)<\alpha$, on the surplus of $X$. For a thorough analysis of the implication of the lack of surplus invariance for acceptability based on Expected Shortfall we refer to Koch-Medina \& Munari~\cite{KM2016}.
\end{example}

%%%%%%%%%%%%%%%%%%%%%%%%%%%%%%%%%%%%%%%%%%%%%%%%%%%%

\subsection{Duality and surplus invariance}

In this section we provide an external characterization of surplus-invariant acceptance sets in $\cX$, by exploiting duality in $\left(L^\infty,\sigma(L^\infty,L^1)\right)$. This ``dual'' representation constitutes the key ingredient in the proof of the main result on the structure of surplus-invariant acceptance sets given later on.

\smallskip

The starting point is the following result highlighting how to recover a surplus invariant acceptance set from its restriction to $L^\infty$ and from its closure in $L^0$. In the sequel, for any $\cA\subset\cX$ we will denote by $\cA^\infty$ the set of all essentially bounded elements of $\cA$, i.e. we set
\[
\cA^\infty := \cA\cap L^\infty\,.
\]

\begin{lemma}
\label{lem:density}
Let $\cA\subset \cX$ be a closed, surplus-invariant acceptance set. Then, the following statements hold:
\begin{enumerate}[(i)]
\item $\cA = \closure_\cX(\cA^\infty)$.
\item $\cA = \closure_{L^0}(\cA)\cap \cX$.
\item If, in addition, $\cA$ is convex, then $\cA^\infty$ is closed with respect to the weak-star topology in $L^\infty$.
\end{enumerate}
\end{lemma}
\begin{proof}
{\em (i)} Clearly, $\closure_\cX(\cA^\infty)\subset \cA$. To prove the converse inclusion take $X\in\cA_-$ and note that $X_n:=\max\{X,-n\}$ belongs to $\cA^\infty$ for every $n\in\N$ and that $X_n\to X$ in $\cX$ by the $\sigma$-Lebesgue property. It follows that $X\in\closure_\cX(\cA^\infty_-)$. The statement now follows by Proposition \ref{proposition: first characterization surpinv}.

\smallskip

{\em (ii)} It is clear that $\cA \subset\closure_{L^0}(\cA)\cap \cX$. To prove the converse assume $X\in \closure_{L^0}(\cA)\cap \cX$. By surplus invariance we may assume without loss of generality that $X\in\cX_-$. Since, by {\em (i)}, we have $\closure_{L^0}(\cA)=\closure_{L^0}(\cA^\infty)$, we find $X_n\in\cA^\infty_-$ such that $X_n\to X$ in $L^0$. By passing to a subsequence, we may assume that $X_n\to X$ a.s. and set for every $n\in\N$
\[
Z_n := \sup_{k\ge n}\{X_k\}\,.
\]
Note that, for every $n$, we have $ X_n\le  Z_n \le 0$ so that $Z_n$ belongs to $L^\infty$ and, by solidity of $\cA^\infty_-$, also to $\cA^\infty_-$. Moreover, $Z_n\downarrow X$ a.s., so that, by the $\sigma$-Lebesgue property, $Z_n\to X$ in $\cX$. Since $\cA$ is closed in $\cX$ we conclude $X\in\cA$.

\smallskip

{\em (iii)} By Proposition~\ref{proposition: first characterization surpinv}, it suffices to prove that $\cA^\infty_-$ is weak-star closed. Take a norm-bounded sequence $(X_n)$ in $\cA^\infty_-$ converging a.s. to $X\in L^\infty$. By the same procedure as above we can construct a monotone decreasing sequence $Z_n\in\cA^\infty$ converging a.s. to $X$. By the $\sigma$-Lebesgue property we infer that $Z_n\to X$ in $L^0$. Hence, $X\in \closure_{L^0}(\cA)\cap \cX$, which by part {\em (ii)}, implies $X\in\cA$. It follows from point {\em (iii)} in Remark~\ref{rem: examples underlying space} that $\cA^\infty$ is weak-star closed.
\end{proof}

\smallskip

We now turn to describe the prototype of a closed, convex, surplus-invariant acceptance set in $\cX$. A map $\varphi:L^\infty_+\to\R\cup\{-\infty\}$ satisfying
\begin{enumerate}[{\rm (F1)}]
\item $\varphi(Z)\leq 0$ for all $Z\in L^\infty_+$,
\item $\varphi(Z)>-\infty$ for some $Z\in L^\infty_+$,
\item $\varphi$ is decreasing,
\end{enumerate}
will be called a {\em decreasing floor function}.

\begin{proposition}
\label{prop: dual rep}
Let $\varphi:L^\infty\to\R\cup\{-\infty\}$ be a decreasing floor function. Then, the set
\begin{equation}
\label{dual repr convex xs acc}
\cA_\cX = \{X\in \cX \,; \ -\E[X^-Z]\geq\varphi(Z), \ \forall Z\in L^\infty_+\}
\end{equation}
is a closed, convex, surplus-invariant acceptance set in $\cX$. Moreover,
\begin{equation}
\label{closures representation}
\cA_\cX =\closure_\cX(\cA_{L^\infty})\,.
\end{equation}
\end{proposition}
\begin{proof}
Clearly $\cA_\cX$ is convex and surplus invariant. It is also clear that $\cA_\cX$ is weak-star closed if $\cX=L^\infty$. If $\cX \neq L^\infty$, we consider a sequence $(X_n)$ in $\cA_\cX$ converging to $X$ in $\cX$ and, therefore, also in $L^0$. If necessary passing to a suitable subsequence of $(X_n)$, Fatou's Lemma now implies that for every $Z\in L^\infty_+$ we have
\[
-\E[X^-Z] \geq -\liminf_{n\to\infty}\E[X_n^-Z] \geq \varphi(Z)\,.
\]
Since $Z$ was arbitrary, we conclude that $X\in\cA_\cX$. The statement in \eqref{closures representation} is a direct consequence of Lemma~\ref{lem:density}.
\end{proof}

\smallskip

We now turn to a dual representation for surplus-invariant acceptance sets in $\cX$ by means of decreasing floor functions. Recall that the \emph{(lower) support function} of a nonempty subset $\cC\subset L^\infty$ is the superlinear and upper-semicontinuous map $\sigma_\cC:L^1\rightarrow\R\cup\{-\infty\}$ defined by
\begin{equation}
\sigma_\cC(Z) := \inf_{X\in\cC}\E[XZ]\,.
\end{equation}
The effective domain of $\sigma_\cC$, called the \emph{barrier cone} of $\cC$, is the convex cone
\[
B(\cC) := \{Z\in L^1 \,; \ \sigma_\cC(Z)>-\infty\}\,.
\]
If $\cC$ is a cone, then $Z\in B(\cC)$ if and only if $\sigma_\cC(Z)=0$. By Lemma~3.11 in Farkas, Koch-Medina \& Munari \cite{FarkasKochMunari2014}, we have $B(\cA)\subset L^1_+$ whenever $\cC$ is a monotone set.

\begin{theorem}
\label{theorem: dual repr convex xs acc}
Let $\cA\subset\cX$ be a closed, convex, surplus-invariant acceptance set. Then, there exists a decreasing floor function $\varphi:L^\infty_+\to\R\cup\{-\infty\}$ such that
\begin{equation}
\label{general dual repr convex xs acc}
\cA = \bigcap_{Z\in L^\infty_+}\{X\in\cX \,; \ -\E[X^-Z]\geq\varphi(Z)\}\,.
\end{equation}
For $\varphi$ we can always choose the restriction of $\sigma_{\cA^\infty}$ to $L^\infty_+$. Moreover, we have $\sigma_{\cA^\infty}\geq\varphi$ on $L^\infty_+$ for every decreasing floor function satisfying~\eqref{general dual repr convex xs acc}.
\end{theorem}
\begin{proof}
By Lemma~\ref{lem:density}, $\cA^\infty$ is a weak-star closed, convex, surplus-invariant acceptance set in $L^\infty$. Moreover, by Theorem~4.1 in Koch-Medina, Moreno-Bromberg \& Munari \cite{KMM2015}, we have
\begin{equation}
\label{dual repr Linfty}
\cA^\infty = \bigcap_{Z\in L^1_+}\{X\in L^\infty \,; \ -\E[X^-Z]\geq\sigma_{\cA^\infty}(Z)\}\,.
\end{equation}
To see that it suffices to take the intersection over $Z\in L^\infty_1$ instead of $Z\in L^1_+$, assume that $X\in L^\infty$ satisfies
\[
-\E[X^-Z]\geq\sigma_{\cA^\infty}(Z)
\]
for every $Z\in L^\infty$. Then, if $Z\in L^1_+$ we may consider $Z_n:=\min\{Z,n\}\in L^\infty_+$ for every $n\in\N$. We then have
\[
-\E[X^-Z_n] \geq \sigma_{\cA^\infty}(Z_n) \geq \sigma_{\cA^\infty}(Z)\,,
\]
where we have used that $\sigma_{\cA^\infty}$ is decreasing. Noting that the sequence $(X^-Z_n)$ is monotone increasing, we can now let $n$ tend to $\infty$ and obtain
\[
-\E[X^-Z] = -\lim_{n\to\infty}\E[X^-Z_n] \geq \sigma_{\cA^\infty}(Z).
\]
It follows that
\[
\cA^\infty = \bigcap_{Z\in L^\infty_+}\{X\in L^\infty \,; \ -\E[X^-Z]\geq\sigma_{\cA^\infty}(Z)\}\,.
\]

\smallskip

Applying Lemma~\ref{lem:density} and Proposition~\ref{prop: dual rep} we now obtain
\[
\cA=\closure_\cX(\cA^\infty)=\bigcap_{Z\in L^\infty_+}\{X\in \cX \,; \ -\E[X^-Z]\geq\sigma_{\cA^\infty}(Z)\}\,.
\]
Since the maximality of $\sigma_{\cA^\infty}$ in \eqref{general dual repr convex xs acc} is clear, this concludes the proof of the theorem.
\end{proof}

\smallskip

\begin{remark}
\begin{enumerate}[(i)]
\item Note that the representation \eqref{general dual repr convex xs acc} involves halfspaces generated by ``functionals'' in $L^\infty$. In particular, if $\cX$ is continuously embedded into $L^1$, this implies that closed, convex, surplus-invariant acceptance sets are automatically $\sigma(\cX,L^\infty)$-closed.
\item The preceding result holds for any space $\cX=L^p$, including the case $p<1$. In spite of the structural lack of local convexity of these spaces, surplus invariance allows to provide an external characterization by using ``duality'' theory in $L^\infty$ equipped with the weak-star topology. In this sense, the representation \eqref{general dual repr convex xs acc} should be compared with the bipolar representation on $L^0$ obtained in Brannath \& Schachermayer \cite{brannath1999} and generalized in Kupper \& Svindland \cite{kupper2011}.
\end{enumerate}
\end{remark}

%%%%%%%%%%%%%%%%%%%%%%%%%%%%%%%%%%%%%%%%%%%%%%

\subsection{The structure of surplus invariance}

In this section we prove the main results of the paper. For any $A\in\cF$ and $\cA\subset\cX$ we define the subset of $\cX$
\[
1_A\cA := \{1_AX \,; \ X\in\cA\}\,.
\]
In particular, we set
\[
\cX(A) := 1_A\cX\,.
\]
The corresponding positive and negative cones are denoted by $\cX_+(A)$ and $\cX_-(A)$, respectively. More precisely, we set
\[
\cX_+(A):=\cX(A)\cap\cX_+ \ \ \ \mbox{and} \ \ \ \cX_-(A):=\cX(A)\cap\cX_-\,.
\]

\smallskip

Recall that the {\em recession cone} of a closed, convex set $\cA\subset\cX$ containing zero is defined by
\[
\rec(\cA) := \bigcap_{\alpha>0}\alpha\cA
\]
and is the largest closed, convex cone contained in $\cA$. We omit the straightforward proof of the following lemma.

\begin{lemma}
\label{recession cone lemma}
Let $\cA\subset\cX$ be a closed, convex, surplus-invariant acceptance set. Then $\rec(\cA)$ is a closed, coherent, surplus-invariant acceptance set.
\end{lemma}

\smallskip

The following result provides a characterization of closed, coherent, surplus-invariant acceptance set in terms of a ``decomposition'' of $\Omega$ into a measurable set $C$ of ``no control'' and a set $C^c$ of ``full control''. Its proof relies on the exhaustion technique used in the proof of Th\'{e}or\`{e}me~2 in Yan~\cite{yan1980}.

\begin{lemma}
\label{thm: coherent structures}
Let $\cA\subset\cX$ be a closed, convex, surplus-invariant acceptance set. Then, there exists a set $C\in\cF$ such that we have
\begin{enumerate}[(i)]
\item $Z=0$ almost surely on $C$ for every $Z\in B(\cA^\infty)$;
\item $Z^\ast>0$ almost surely on $C^c$ for some $Z^*\in B(\cA^\infty)\cap L^\infty$.
\end{enumerate}
Moreover,
\[
\cA=\cX(C)\oplus 1_{C^c}\cA\quad\mbox{ and }\quad\rec(\cA)=\cX(C)\oplus \cX_+(C^c)\,.
\]
\end{lemma}
\begin{proof}
We first prove that the class
\[
\cG := \{\{Z=0\} \,; \ Z\in B(\cA^\infty)\}
\]
is closed under countable intersections. Indeed, consider a sequence $(Z_n)$ in $B(\cA^\infty)\cap L^\infty$ and take a sequence $(\alpha_n)$ of positive real numbers such that
\[
\sum_{n\in\N}\alpha_n \lnorm Z_n\rnorm_\infty \ \ \ \mbox{and} \ \ \ \sum_{n\in\N} \alpha_n\sigma_{\cA^\infty}(Z_n)
\]
both converge. The first condition implies that $\sum_n\alpha_n Z_n$ converges to some $Z$ in the norm topology of $L^\infty$ and the second, by the upper semicontinuity of $\sigma_{\cA^\infty}$, that $Z\in B(\cA^\infty)$. Moreover,
\[
\{Z=0\} = \bigcap_{n\in\N}\{Z_n=0\}\,.
\]
It follows that $\cG$ is closed under countable intersections. Take now a sequence $(Z_n)$ in $B(\cA^\infty)$ such that
\[
\lim_{n\to\infty}\probp(Z_n=0) = \inf_{E\in\cG}\probp(E)\,.
\]
Then, we find a suitable $Z^\ast\in B(\cA^\infty)$ such that
\[
\{Z^\ast=0\} = \bigcap_{n\in\N}\{Z_n=0\}\,.
\]
In particular, $\{Z^\ast=0\}$ attains the minimal probability over $\cG$. We claim that
\[
C:=\{Z^\ast=0\}
\]
has the desired properties.

\smallskip

Clearly, $Z^\ast$ satisfies the condition in {\em (ii)}. To prove that {\em (i)} holds, assume there exists $Z\in B(\cA^\infty)$ with $Z>0$ on a measurable subset $E$ of $C$ with nonzero probability. Then, $Z+Z^\ast$ would be an element of $B(\cA^\infty)$ satisfying $\probp(Z+Z^\ast=0)\leq\probp(C\setminus E)<\probp(C)$, contradicting the minimality of $C$. Hence, {\em (i)} is satisfied.

\smallskip

Now, take $X\in\cX$. Then, for every $Z\in B(\cA^\infty)$ we have $\E[XZ]=\E[1_{C^c}XZ]$, which implies
\[
\E[XZ]\ge \sigma_{\cA}(Z) \ \iff \ \E[1_{C^c}XZ]\ge \sigma_{\cA}(Z)\,.
\]
By the dual representation in Theorem~\ref{theorem: dual repr convex xs acc}, this implies that $\cA=\cX(C)\oplus1_{C^c}\cA$. Moreover,  since $\rec(\cA)\subset\cA$ we have that $B(\cA^\infty)\subset B(\rec(\cA)^\infty)$. Take $X\in\rec(\cA)$, then by surplus invariance $-X^-\in\rec(\cA)$ and
\[
-\E[1_{C^c}X^-Z]\ge \sigma_{\rec(\cA)^\infty}(Z)=0
\]
where $\sigma_{\rec(\cA)^\infty}(Z)=0$ holds since $\rec(\cA)^\infty$ is a cone. But this is only possible if $-1_{C^c}X^- =0$. This implies $\rec(\cA)=\cX(C)\oplus \cX_+(C^c)$.
\end{proof}

\smallskip

\begin{theorem}
\label{main theorem: coherent case}
Let $\cA\subset\cX$ be a closed, coherent acceptance set. Then, $\cA$ is surplus invariant if and only if there exists $A\in\cF$ with $\probp(A)>0$ such that
\[
\cA = \cX_+(A)\oplus\cX(A^c) = \{X\in\cX \,; \ 1_AX\geq 0\}\,.
\]
\end{theorem}
\begin{proof}
The ``if'' implication is clear. To prove the ``only if'' implication, note that $\cA=\Rec(\cA)$ because $\cA$ is a convex cone. Hence, the claim follows immediately from the preceding lemma once we set $A=C^c$.
\end{proof}

\smallskip

\begin{remark}
The above result extends the characterization of weak-star closed, coherent, surplus-invariant acceptance sets obtained in Koch-Medina, Moreno-Bromberg \& Munari \cite{KMM2015} in the context of spaces of bounded random variables. Note that, in contrast to \cite{KMM2015}, the separability of $\cX$ is no longer required.
\end{remark}

\smallskip

The preceding lemma allows us to derive a structural result for general closed, convex, surplus-invariant acceptance sets. Recall that a set $\cD\subset\cX$ is {\em tight}, or {\em bounded in probability}, if for every $\e\in(0,1)$ there exists $M>0$ such that $\probp(-M<X<M)>\e$ for all $X\in\cD$. In other words, tight sets are precisely those sets which are topologically bounded in $L^0$.

\begin{theorem}
\label{main theorem}
Let $\cA\subset\cX$ be a closed, convex acceptance set. Then, $\cA$ is surplus invariant if and only if there exists a measurable partition $\{A,B,C\}$ of $\Omega$ with $\probp(C)<1$, unique up to modifications on sets of nonzero probability, such that
\begin{equation}
\label{representation formula main theorem}
\cA = \cX_+(A)\oplus (\cX_+(B)+\cD_B)\oplus\cX(C)\,,
\end{equation}
where $\cD_B$ is a closed, convex, solid, tight subset of $\cX_-(B)$ such that for every measurable set $E\subset B$ satisfying $\probp(E)>0$ there exists $X\in\cD_B$ with $\probp(1_EX<0)>0$.
\end{theorem}
\begin{proof}
The ``if'' implication is clear, hence we focus on the converse statement. To prove the ``only if'' implication we will define the partition explicitly. By Lemma~\ref{thm: coherent structures} we find $Z^*\in B(\cA^\infty)\cap L^\infty$ so that
\[
\rec(\cA)=\cX(C)\oplus\cX_+(C^c)\,,
\]
where $C=\{ Z^\ast =0\}$. It is clear that $C^c=\{V=0\}$ where
\[
V=\esssup_{X\in\rec(\cA)}X^-\,.
\]
Set now
\[
A:=\{U=0\} , \ \ \ B:=\{V=0\}\setminus\{U=0\}\,,
\]
where
\[
U=\esssup_{X\in\cA}X^-\,.
\]
Then, it follows that $\{A,B,C\}$ is a measurable partition of $\Omega$ satisfying
\[
\cA = 1_A\cA\oplus1_B\cA\oplus1_C\cA = \cX_+(A)\oplus1_B\cA\oplus \cX(C)\,.
\]
From the definition of $B$ it is immediate to see that, whenever $\probp(B)>0$ and $E$ is a measurable subset of $B$ with strictly positive probability, we always find $X\in\cD_B$ such that $\probp(1_EX<0)>0$.

\smallskip

To conclude the proof, we only need to show that the closed, convex, solid set $\cD_B=1_B\cA_-$ is tight. To this effect, assume $\cD_B$ is not tight so that we find $\varepsilon\in(0,1)$ such that for every $M>0$ there exists $X_M\in\cD_B$ with $\probp(X_M<-M)>\varepsilon$. Noting that $\{X_M<-M\}\subset B$ and recalling that $C^c=\{Z^\ast>0\}$, we immediately obtain
\[
-M\E[1_{\{X_M<-M\}}Z^\ast] \geq \E[X_MZ^\ast] \geq \sigma_{\cA^\infty}(Z^\ast)>-\infty
\]
for any $M>0$, which is impossible since the expectation on the left-hand side is strictly positive by the assumption that $\probp(X_M<-M)>\varepsilon$. It follows that $\cD_B$ is tight and this concludes the proof.
\end{proof}

\smallskip

\begin{remark}
The last condition stated in the preceding theorem, namely that for every measurable set $E\subset B$ satisfying $\probp(E)>0$ there exists $X\in\cD_B$ with $\probp(1_EX<0)>0$, ensures the uniqueness of the above decomposition. In particular, this ``sensitivity'' condition implies that the event $A$ is the maximal measurable subset of $\Omega$ on which acceptable positions are required to be positive.
\end{remark}

\smallskip

The preceding results have interesting financial implications. Indeed, acceptability can be described by requirements on the behaviour of capital positions on each of the ``atoms'' of a measurable partition $\{A,B,C\}$ of $\Omega$: no defaults are allowed on $A$, a ``controlled'' default is allowed on $B$, and no requirements are imposed on $C$. The second condition will be made precise in Section \ref{sect: stochastic boundedness} by elaborating on the notion of tightness.

\smallskip

We illustrate the structure of surplus invariance by applying our previous result to the concrete examples of convex, surplus-invariant acceptance sets introduced above.

\begin{example}[Shortfall Risk]
\label{ex: shortfall risk decomposition}
Assume $\cX$ is continuously embedded into $L^1$. For any loss function $\ell:\R_+\to\R_+$ and $c>0$ the surplus-invariant acceptance set
\[
\cA = \{X\in\cX \,; \ \E[\ell(X^-)]\leq c\}
\]
is closed and convex. In this case, we have $\probp(B)=1$.

\smallskip

To prove the statement, we will first show that $\Rec(\cA)=\cX_+$. The inclusion ``$\supset$'' is clear. To prove the converse inclusion, take $X\in\Rec(\cA)$ and assume that $X\notin\cX_+$ so that $\probp(X\leq-\e)>0$ for some $\e>0$. Moreover, take $\gamma>0$ satisfying $\ell(\gamma\e)>0$. Since $X\in\Rec(\cA)$, we need to have
\[
c \geq \E[\ell((\alpha\gamma X)^-)] \geq \alpha\E[\gamma1_{\{X\leq-\e\}}\ell(X^-)] \geq \alpha\ell(\gamma\e)\probp(X\leq-\e)
\]
for every $\alpha>1$ by convexity, which is clearly not possible. As a result, $\Rec(\cA)=\cX_+$ holds.

\smallskip

Now, since $\Rec(\cA)=1_{B^c}\cA+1_B\cX_+$ as a consequence of the above decomposition theorem and since $\cA$ clearly contains positions that are strictly negative almost surely, such as $X=-\e1_\Omega$ for $0<\e<c$, we conclude that $\probp(B^c)=0$ must hold so that $\probp(B)=1$.
\end{example}

\smallskip

\begin{example}[Test Scenarios]
Consider a set $E\in\cF$ satisfying $\probp(E)>0$. The surplus-invariant acceptance set
\[
\cA = \{X\in\cX \,; \ 1_EX\geq0\}
\]
is closed and coherent. Since $\cA=\cX_+(E)$, we easily see that we can decompose $\cA$ by taking $A=E$ and $B=C=\emptyset$.
\end{example}

\smallskip

\begin{example}[Expected Tail Loss]
\label{ex: expected tail loss decomposition}
Assume $\cX$ is continuously embedded into $L^1$. For any $\alpha\in(0,1)$ and $c\geq0$ the surplus-invariant acceptance set
\[
\cA = \{X\in\cX \,; \ \ES_\alpha(-X^-)\leq c\}
\]
is closed and convex. Moreover, it is coherent if and only if $c=0$, in which case $\cA=\cX_+$. We claim that, if $c>0$, then $\probp(B)=1$.

\smallskip

To see this, note first that we must have $\Rec(\cA)=\cX_+$ by positive homogeneity. Since $\Rec(\cA)=1_{B^c}\cA+1_B\cX_+$ as a consequence of the above decomposition theorem and since $\cA$ clearly contains positions that are strictly negative almost surely, such as $X=-\e1_\Omega$ for $0<\e<c$, we conclude that $\probp(B^c)=0$ must hold so that $\probp(B)=1$.
\end{example}

\smallskip

We conclude by showing that, with the clear exception of atomic spaces with a finite number of atoms, the decomposition obtained in Theorem \ref{main theorem} does not hold if we equip $L^\infty$ with the norm topology.

\begin{example}[$L^\infty$ with the strong topology]
Let $(\Omega,\cF,\probp)$ be a probability space admitting an infinite partition $(B_n)$ of $\Omega$ consisting of measurable sets of nonzero probability. This is equivalent to $L^\infty$ being infinite dimensional. Define the increasing sequence $(A_n)$ by setting
\[
A_n := \bigcup_{k=1}^n B_k
\]
for every $n\in\N$. By construction we have $\probp(A_{n+1}\setminus A_n)=\probp(B_{n+1})>0$ for each $n\in\N$ and $\bigcup_{n}A_n=\Omega$. For each $n\in\N$ the set
\[
\cD_n := L^\infty_-(A_n)
 \]
is norm-closed, convex and solid. Moreover, note that $(\cD_n)$ is increasing. Hence, the union
\[
\cD := \bigcup_{n\in\N}\cD_n
\]
is easily seen to be a convex, solid subset of $L^\infty_-$. It is also clear that $\cD$ is not tight. Now, denote by $\overline{\cD}$ the closure of $\cD$ in the norm topology and consider the closed, convex, surplus invariant, monotone set
\[
\cA := \overline{\cD}+L^\infty_+\,.
\]
Since $\probp(X\ge -\frac{1}{2})>0$ holds for every $X\in\overline{\cD}$, we see that $\overline{\cD}\neq L^\infty_-$. In particular, $\cA$ is an acceptance set. We claim that $\cA$ cannot be decomposed as in Theorem~\ref{main theorem}. For assume we could write $\cA$ in the standard form \eqref{representation formula main theorem}. In this case, we would have $\overline{\cD}=\cD_B\oplus L^\infty_-(C)$, where $\cD_B$ is tight. Since $1_B\overline{\cD}$ is not tight, we must have $\probp(B)=0$ so that $\overline{\cD}=L^\infty_-(C)$. This would imply that $-1_{A_n}\in L^\infty_-(C)$, hence $\probp(A_n)\leq\probp(C)$, for any $n\in\N$. However, this would be possible only if $\probp(C)=1$, in contradiction to $\overline{\cD}\neq L^\infty_-$.
\end{example}

%%%%%%%%%%%%%%%%%%%%%%%%%%%%%%%%%%%%

\subsection{Stochastic boundedness and surplus invariance}
\label{sect: stochastic boundedness}

In the previous section we have proved that any closed, convex, surplus-invariant acceptance set $\cA\subset\cX$ can be decomposed as
\[
\cA = \cX_+(A)\oplus (\cX_+(B)+\cD_B)\oplus\cX(C)
\]
for a suitable partition $\{A,B,C\}$ of $\Omega$. In particular, the set $\cD_B$ consists of acceptable default options and was shown to be tight, or bounded in probability. In this section we focus on the set $\cD_B$ and show how to interpret this ``boundedness'' property from a capital adequacy perspective.

\smallskip

The distribution function of a random variable $X\in L^0$ will be denoted by $F_X$, i.e. we set $F_X(t):=\probp(X\le t)$ for every $t\in\R$. Recall that if $X$ and $Y$ are two random variables, $X$ is said to be {\em (first order) stochastically preferred to} $Y$, denoted by $X \stle Y$, if $F_X(t)\leq F_Y(t)$ for every $t\in\R$. A set $\cD\subset L^0$ is {\em (first order) stochastically bounded below} by a random variable $X^\ast$ if every element of $\cD$ is (first order) stochastically preferred to $X^\ast$, i.e. if $X\stle X^\ast$ for all $X\in\cD$.

\smallskip

\begin{remark}
Note that the above notions of stochastic preference and stochastic boundedness depend only on the distribution of the involved random variables. In particular, as we will do later in the context of some examples, one need not assume that a stochastic bound for a set $\cD\subset L^0$ is defined on our fixed probability space $(\Omega,\cF,\probp)$.
\end{remark}

\smallskip

The following result establishes that, for a subset of $L^0_-$, tightness is equivalent to stochastic boundedness.

\begin{proposition}
\label{prop: tight and stoch bounded}
A set $\cD\subset L^0_-$ is tight if and only if it is stochastically bounded below.
\end{proposition}
\begin{proof}
We show first the ``if'' implication. Assume $X^\ast\in L^0$ is a stochastic bound for $\cD$. In particular, note that $X^\ast$ must be negative. Now, for any $\e\in(0,1)$ we find $M>0$ large enough to satisfy $\probp(X^\ast>-M)>\e$. Since $\probp(X>-M)\geq\probp(X^\ast>-M)$ for all $X\in\cD$, we conclude that $\cD$ is tight.

\smallskip

To prove the converse implication, assume $\cD$ is tight and consider the function $F:\R_+\to[0,1]$ defined by setting
\[
G(x) := \sup_{X\in\cD}F_X(-x)\,.
\]
It is clear that $G$ is decreasing and satisfies
\begin{equation}\label{tightness consequence}
\lim_{x\to\infty}G(x) = 0\,.
\end{equation}
The last assertion follows directly from the tightness of $\cD$. The proof now reduces to showing that there exists a random variable $X^\ast\in L^0_-$ satisfying
\begin{equation}
\label{eq: goal stochastic bound}
F_{X^\ast}(-x)\geq G(x) \ \ \mbox{for every} \ x>0\,.
\end{equation}
Assume first that there exists $M>0$ such that $G(M)=0$, i.e. the set $\cD$ is uniformly bounded in $L^\infty$. Then, the random variable $X^\ast:=-M1_\Omega$ is easily seen to be a stochastic bound for $\cD$.

\smallskip

Assume next that for every $M>0$ we have $G(M)>0$. Together with~\eqref{tightness consequence}, this implies that for every $\e\in(0,1)$ there exists $M>0$ such that
\[
0 < \sup_{X\in\cD}F_X(-M) \leq \e\,.
\]
In this case, it is not difficult to show that we can find a countable partition $(A_n)$ of $\Omega$ consisting of measurable sets with nonzero probability. Now, consider an increasing sequence $(\alpha_n)$ of strictly positive numbers such that $G(\alpha_n)<1-\sum_{k=1}^n \probp(A_k)$ and set
\[
X^\ast := -\sum_{n=1}^\infty \alpha_n\indic_{A_n}\,.
\]
We claim that \eqref{eq: goal stochastic bound} holds. To this end, take $x>0$ and choose the smallest $n_0\in\N$ such that $G(x)\geq 1-\sum_{k=1}^{n_0}\probp(A_k)$. Since $G$ is decreasing, we have $\alpha_n\geq x$ for each $n\geq n_0$. Thus
\[
F_{X^\ast}(-x) \geq \sum_{k=n_0}^\infty \probp(A_k) = 1-\sum_{k=1}^{n_0-1}\probp(A_k) > G(x)\,,
\]
proving \eqref{eq: goal stochastic bound} and concluding the proof of the proposition.
\end{proof}

\smallskip

\begin{remark}
One could also prove that the function $G$ defined in the above proof is right continuous and infer the existence of a random variable on a nonatomic probability space having $1-G$ as its distribution function. We prefer the above proof because it does not require us to consider random variables on a probability space that is different from our underlying probability space $(\Omega,\cF,\probp)$.
\end{remark}

\smallskip

We proceed to apply the preceding characterization of tightness to complement the representation result obtained in Theorem \ref{main theorem}. To this end, recall that the {\em Value-at-Risk} ($\VaR$) of $X\in L^0$ at the level $\alpha\in(0,1)$ is defined as
\[
\VaR_\alpha(X) := \inf\{t\in\R \,; \ \probp(X+t<0)\leq\alpha\}\,.
\]
Moreover, the {\em Expected Shortfall} ($\ES$) of $X$ at level $\alpha$ is given by
\[
\ES_\alpha(X) := \frac{1}{\alpha}\int_0^\alpha\VaR_\beta(X)d\beta\,.
\]

\smallskip

The first part of the following simple lemma recalls a well-known characterization of first order stochastic dominance. The second part is an obvious corollary and is an equivalent formulation that first order stochastic dominance implies second order stochastic dominance.

\begin{lemma}
\label{lemma: stochastic dominance and var}
For any random variables $X,Y\in L^0$, we have
\[
X\stle Y \ \iff \ \VaR_\alpha(X)\le \VaR_\alpha(Y) \ \ \mbox{for every} \ \alpha\in(0,1)\,.
\]
In particular,
\[
X\stle Y \ \implies \ \ES_\alpha(X)\le \ES_\alpha(Y) \ \ \mbox{for every} \ \alpha\in(0,1)\,.
\]
\end{lemma}

\smallskip

The preceding lemma allows us to link the ``boundedness'' of the set of acceptable default profiles $\cD_B$ to a stringent control of the $\VaR_\alpha$ and $\ES_\alpha$ of its elements for arbitrary $\alpha\in(0,1)$.

\begin{theorem}
Let $\cA\subset \cX$ be a closed, convex, surplus-invariant acceptance set with standard decomposition
\[
\cA = \cX_+(A)\oplus (\cX_+(B)+\cD_B)\oplus\cX(C)\,.
\]
Then, there exists $X^\ast\in L^0_-$ such that
\[
\VaR_\alpha(\indic_B X)\leq\VaR_\alpha(X^\ast) \ \ \mbox{for every} \ X\in\cA, \ \alpha\in(0,1)\,.
\]
In particular, if $X^\ast\in L^1_-$ then
\[
\ES_\alpha(\indic_B X)\leq\ES_\alpha(X^\ast), \ \ \mbox{for every} \ X\in\cA, \ \alpha\in(0,1)\,.
\]
\end{theorem}
\begin{proof}
By Theorem \ref{main theorem}, the set $\cD_B$ is a tight subset of $\cX_-(B)$. Hence, Proposition \ref{prop: tight and stoch bounded} implies that $\cD_B$ is stochastically bounded by some $X^\ast\in L^0_-$. The two assertions follow immediately from Lemma~\ref{lemma: stochastic dominance and var}.
\end{proof}

\smallskip

The above result is also important since it provides a hint on how to {\em construct} closed, convex, surplus-invariant acceptance sets that allow ``controlled'' defaults but have no ``uncontrolled'' scenarios.

\begin{example}
\label{cor: construction of surplus invariance via ES}
Assume $\cX$ is continuously embedded into $L^1$ and consider a measurable partition $\{A,B\}$ of $\Omega$. Moreover, take $X^\ast\in \cX_-$. Then, the set
\[
\cA = \cX_+(A)\oplus (\cX_+(B)+\cD_B)
\]
where
\[
\cD_B = \{X\in \cX_-(B) \,; \ \ES_\alpha(X)\leq\ES_\alpha(X^\ast), \ \forall \alpha\in(0,1)\}
\]
is a closed, convex, surplus-invariant acceptance set in $\cX$.

\smallskip

Indeed, it is enough to prove that $\cD_B$ is closed, convex and tight in $\cX$. Since $\ES_\alpha$ is continuous and convex on $L^1$, the set $\cD_B$ is clearly closed and convex. To prove tightness, recall that $\ES_\alpha(X)\to\E[-X]$ as $\alpha\to1$ for any $X\in L^1$. As a result, it follows that $\cD_B$ is bounded in $L^1$ and hence, {\em a fortiori}, in $L^0$.
\end{example}

We conclude this section by illustrating the preceding result in the context of the concrete examples of convex (noncoherent), surplus-invariant acceptance sets discussed above.

\begin{example}[Shortfall Risk]
Assume $\cX$ is continuously embedded into $L^1$. Consider a loss function $\ell:\R_+\to\R_+$ and take $c>0$. It follows from Example~\ref{ex: shortfall risk decomposition} that the closed, convex, surplus-invariant acceptance set
\[
\cA = \{X\in\cX \,; \ \E[\ell(X^-)]\leq c\}
\]
satisfies $\probp(B)=1$, so that $\cA_-$ is tight. We want to exhibit a stochastic bound for $\cA_-$.

\smallskip

Note that, for every $X\in\cA_-$, we have
\[
c\geq\E[\ell(X^-)]\geq \E[1_{\{X\leq x\}}\ell(-x)] = \ell(-x)\probp(X\leq x)=\ell(-x)F_X(x),
\]
where $x\leq 0$ is arbitrary. Now, choose $x^\ast<0$ such that $\ell(-x^\ast)=c$ and define a function $F:\R\to[0,1]$ by setting
\[
F(x) =
\begin{cases}
\frac{c}{\ell(-x)} & \mbox{if} \ x\leq x^\ast,\\
1 & \mbox{if} \ x>x^\ast.
\end{cases}
\]
Then, any random variable with distribution function $F$ is easily seen to be a stochastic bound (possibly defined on a different probability space) for $\cA_-$.
\end{example}

\smallskip

\begin{example}[Expected Tail Loss]
Assume $\cX$ is continuously embedded into $L^1$. Take $\alpha\in(0,1)$ and $c>0$. In Example \ref{ex: expected tail loss decomposition} we saw that the closed, convex, surplus-invariant acceptance set
\[
\cA = \{X\in\cX \,; \ \ES_\alpha(-X^-)\leq c\}
\]
satisfies $\probp(B)=1$, so that $\cA_-$ is tight. We want to exhibit a stochastic bound for $\cA_-$.

\smallskip

First, recall from Theorem~4.52 and Remark~4.53 in F\"{o}llmer \& Schied~\cite{FoellmerSchied2011} that we can write
\[
\ES_\alpha(X) = \max_{Z\in \mathcal{Q}_\alpha}\E[-XZ]
\]
for every $X\in\cX$, where the class of random variables over which we are taking the maximum is given by
\[
\mathcal{Q}_\alpha = \Big\{Z\in \cX_+ \,; \ \|Z\|_\infty\leq\frac1\alpha, \ \E[Z]=1\Big\}\,.
\]
In fact, the maximum is attained by the random variable $Z^\ast = \frac1\alpha(1_{\{X<q\}} + \kappa1_{\{X=q\}})$, where $q$ is an $\alpha$-quantile of $X$ and $\kappa\geq0$ is chosen in such a way that $\E[Z]=1$. Now, take an arbitrary $X\in\cA_-$. Since
\[
\alpha c \geq \alpha\ES_\alpha(X) = \E[-X(1_{X<q} + \kappa1_{X=q})]\geq -q\probp[X<q]-q\kappa\probp[X=q] = -\alpha q
\]
by definition of an $\alpha$-quantile, it follows that $q\geq -c$ holds. Hence, for every $x<-c$ we obtain
\[
\alpha c \geq \alpha\ES_\alpha(X) \geq \E[-X1_{X<q}] \geq \E[-X1_{X\leq x}] \geq -x\probp(X\leq x)\,.
\]
Finally, define a function $F:\R\to[0,1]$ by setting
\[
F(x) =
\begin{cases}
-\frac{\alpha c}{x} & \mbox{if} \ x<-c,\\
1 & \mbox{if} \ x\geq-c.
\end{cases}
\]
Then, any random variable with distribution function $F$ is easily seen to be a stochastic bound (possibly defined on a different probability space) for $\cA_-$.
\end{example}

%%%%%%%%%%%%%%%%%%%%%%%%%%%%%%%%%%%%%%

\section{Num\'{e}raire invariance}

Coherent surplus-invariant acceptance sets are related to the concept of num\'{e}raire invariance discussed next. Throughout this section we assume that
\begin{equation}
\label{module condition}
R\in L^\infty, \ \ X\in\cX \ \implies \ RX\in\cX\,.
\end{equation}
Note that this condition implies the localization condition \eqref{localization} assumed thus far.

\smallskip

Recall that we had expressed capital positions in a fixed unit of account, say some fixed currency. The process of changing the accounting currency can be described by means of a random variable $R$ that is strictly positive almost surely, representing the exchange rate from the original into the new currency. For a capital position $X\in \cX$ expressed in the original currency, the random variable $RX$ represents the position of the {\em same} company expressed in the new currency. By our assumption \eqref{module condition} the product $RX$ still belongs to $\cX$. Assume $\cA\subset\cX$ is a capital adequacy test in the original currency, then if $Y$ is a capital position expressed in the new currency, then $Y$ is acceptable if and only if $\frac{Y}{R}\in\cA$, which is the case if and only if
\[
Y\in R\cA\,.
\]
Hence, the particular form of the test changes whenever the currency changes even though the capital positions still belong to $\cX$. This translation step needs to be borne in mind when applying a capital adequacy test across various jurisdictions in their respective currency. This is particulary important when trying to harmonize solvency regimes across the globe. However, discussions about harmonization typically revolve around the choice of one test --- e.g. based on Value-at-Risk or Expected Shortfall --- that is applied across jurisdictions without executing the translation step. If the translation step is not performed, then the only way that the test be consistent across jurisdictions --- i.e. that exactly the same capital positions be accepted in all jurisdictions --- is to require {\em num\'{e}raire invariance}.

\smallskip

For convenience, a random variable $R\in L^\infty_+$ which is strictly positive almost surely will be called a {\em rescaling factor}. Thus, rescaling factors are random variables that qualify to represent a change of unit of account.

\begin{definition}
An acceptance set $\cA\subset \cX$ is num\'{e}raire invariant if for every rescaling factor $R$ we have
\begin{equation}
\label{rescaling condition}
X\in\cA \ \implies \ RX\in\cA\,.
\end{equation}
\end{definition}

\smallskip

We start by showing that we can use a smaller or larger class of rescaling factors without changing the concept of num\'{e}raire invariance. Recall that a random variable $R\in L^\infty_+$ is said to be {\em bounded away from zero} if $\probp(R>\varepsilon)=1$ for some $\varepsilon>0$.

\begin{lemma}
\label{numeraire invariance version}
Assume $\cA\subset \cX$ is a closed acceptance set. Then, $\cA$ is num\'{e}raire invariant if and only if any of the following equivalent conditions are satisfied:
\begin{enumerate}[(a)]
  \item $X\in\cA$ implies $RX\in\cA$ for any $R\in L^\infty_+$;
  \item $X\in\cA$ implies $RX\in\cA$ for any $R\in L^\infty_+$ that is bounded away from zero.
\end{enumerate}
\end{lemma}
\begin{proof}
Clearly, we only need to prove that {\em (b)} implies {\em (a)}. To this end, fix $X\in\cA$ and take any $R\in L^\infty_+$. Moreover, consider the sequence with general term $R_n:= R+\frac{1}{n}\indic_\Omega$. Since $R_n$ is bounded away from zero for any $n\in\N$, we have $R_n X\in\cA$. Since $\cA$ is closed, the limit $RX$ must also belong to $\cA$, proving that {\em (a)} is satisfied.
\end{proof}

\smallskip

\begin{remark}
The main reason for choosing a rescaling factor $R$ to be an element of $L^\infty$ is that for any capital position $X\in \cX$ we again have $RX\in \cX$. However, the following statement is true: a set $\cA\subset \cX$ is num\'{e}raire invariant if and only if for every $R\in L^0_+$ we have
\[
X\in\cA, \ \ RX\in\cX \ \implies \ RX\in\cA\,.
\]
Indeed this property, obviously, implies num\'{e}raire invariance, since for $R\in L^\infty_+$ we always have $RX\in \cX$, hence (a) in Lemma~\ref{numeraire invariance version} is satisfied. To show the converse, note that for $R_n = \min\{R,n\}$ we have that $R_nX\in\cA$ by num\'{e}raire invariance and that $R_nX\rightarrow RX$ in $\cX$ by the $\sigma$-Lebesgue property.
\end{remark}

\smallskip

As a result of the preceding lemma, every num\'{e}raire invariant acceptance set is also surplus invariant. Moreover, these two properties are equivalent for any conical acceptance set.

\begin{proposition}
\label{proposition: numinv is equiv to supinv}
Assume $\cA\subset \cX$ is a closed acceptance set. Then, the following statements are equivalent:
\begin{enumerate}[(a)]
  \item $\cA$ is num\'{e}raire invariant;
  \item $\cA$ is conical and surplus invariant.
\end{enumerate}
\end{proposition}
\begin{proof}
Assume first that $\cA$ is num\'{e}raire invariant. That $\cA$ is a cone is clear from the definition of num\'{e}raire invariance. Now, take $X\in\cA$. Then, the previous lemma yields $X1_{X<0}\in\cA$, which implies surplus invariance. Assume now that $\cA$ is conical and surplus invariant. Take $X\in\cA$ and recall that $-X^-\in\cA$. For any rescaling factor $R$ we have $-RX^-\ge -\lnorm R\rnorm_\infty X^-$. As a consequence of conicity and monotonicity, it follows that $-RX^-\in\cA$. Hence $RX\in\cA$, proving that $\cA$ is num\'{e}raire invariant.
\end{proof}

\smallskip

In light of the close link between num\'{e}raire and surplus invariance, we can use the preceding results to obtain the following representation of num\'{e}raire invariant acceptance sets.

\begin{theorem}
\label{corollary: coherent numinv}
Assume $\cA\subset\cX$ is a closed, convex acceptance set. Then, $\cA$ is num\'{e}raire invariant if and only if there exists $A\in\cF$ such that $\probp(A)>0$ and
\[
\cA = \cX_+(A)\oplus \cX(A^c)\,.
\]
\end{theorem}
\begin{proof}
By Proposition \ref{proposition: numinv is equiv to supinv}, the acceptance set $\cA$ is surplus invariant and conical, hence coherent. Therefore, the claim follows at once from Theorem~\ref{main theorem: coherent case}.
\end{proof}

%%%%%%%%%%%%%%%%%%%%%%%%%%%%%%%%%%%%%%%%%%%%

\section{Final comments}

The paper has undertaken an analysis of very specific aspects --- surplus invariance and num\'{e}raire invariance --- of the simple classical model which, despite its rudimentary nature, underlies the bulk of the risk measures literature in a one-period setting. This model identifies a financial institutions with its net capital position --- assets less liabilities --- and has the feature that any surplus accrues to the owners of the institution and any costs of default are fully borne by liability holders. As a result, liability holders can only experience one of the following two fates:
\begin{enumerate}[(a)]
\item the institution has sufficient funds to honour its obligations, in which case the liability holders are paid what they are owed (but nothing more); or,
\item the institution has a shortfall of funds and defaults on its liabilities so that liability holders are paid either nothing or only part of what is owed to them.
\end{enumerate}
Within this particular model, the only benefit accruing to liability holders is the amount they get paid at maturity. It follows that, in this setting, surplus invariance is a reasonable property to ask from the perspective of liability holders. The structural results on surplus invariance are nothing but the logical consequence of accepting the premises of the model and the financial contribution of our paper is mainly to highlight specific implications of the classical model. In the real world, however, there may be potential benefits to liability holders (or society at large) that may be attributable to institutions running a surplus. Hence, one should be cautious when deriving practical implications from our results. But there are important theoretical implications. Indeed, our results should prompt a questioning of the classical model's premises and should motivate the search for alternative models that have the ability to capture these potential benefits. In this case, one should find specific and plausible economic reasons to amend the model in one or the other direction. One possibility would be to consider a multi-period model which explicitly accounts for benefits to liability holders (or society at large) that are attributable to the sustainability of a financial institution and which can, therefore, be argued to be linked to profitability and, hence, to surplus. Alternatively, one could also question whether equating a company with its capital positions (assets less liabilities) is the most reasonable course of action to take. Indeed, one may want to explore acceptability criteria that operate on assets-liability pairs $(A,L)$ in a more subtle way than just through the net capital position
\[
X = A-L\,.
\]
We hope that our results may generate interest in this line of investigation.

%%%%%%%%%%%%%%%%%%%%%%%%%%%%%%%%%%%%%%%%%%%%%%%%%%%%%%%%%%%%%%%%%


\begin{thebibliography}{99}
\bibliographystyle{plain}

\bibitem{AliprantisBorder2006} Ch. D. Aliprantis, K. C. Border: {\it Infinite dimensional analysis: A hitchhiker's guide}, Third Edition, Springer (2006)

\bibitem{AliprantisBurkinshaw2003} Ch. D. Aliprantis, O. Burkinshaw: {\it Locally solid Riesz spaces with applications to economics}, Second Edition, American Mathematical Society (2003)

\bibitem{ADEH1999} P. Artzner, F. Delbaen, J. M. Eber, D. Heath: Coherent Measures of Risk, {\it Mathematical Finance}, 9(3), 203--228 (1999)

\bibitem{brannath1999} W. Brannath, W. Schachermayer: A bipolar theorem for {$L_+^0(\Omega,\mathcal{F},\probp)$}, In: {\it S\'eminaire de Probabilit\'es XXXIII}, 349--354, Springer (1999)

\bibitem{ContDeguestHe2013} R. Cont, R. Deguest, X. D. He: Loss-based risk measures, {\it Statistics \& Risk Modeling with Applications in Finance and Insurance}, 30(2), 133--167 (2013)

\bibitem{FarkasKochMunari2014} W. Farkas, P. Koch-Medina, C. Munari: Beyond cash-additive risk measures: when changing the num\'eraire fails, {\it Finance and Stochastics}, 18(1), 145--173 (2014)

\bibitem{FoellmerSchied2002} H. F\"ollmer, A. Schied: Convex measures of risk and trading constraints, {\it Finance and Stochastics}, 6(4), 429--447 (2002)

\bibitem{FoellmerSchied2011} H. F\"ollmer, A. Schied: {\it Stochastic finance: An introduction in discrete time}, Walter de Gruyter (2011)

\bibitem{FRG2002} M. Frittelli, E. Rosazza-Gianin: Putting order in risk measures, {\it Journal of Banking \& Finance}, 26(7), 1473--1486 (2002)

\bibitem{Grothendieck1973} A. Grothendieck: {\it Topological vector spaces}, Gordon and Breach (1973)

\bibitem{KMM2015} P. Koch-Medina, S. Moreno-Bromberg, C. Munari: Capital adequacy tests and limited liability of financial institutions, {\it Journal of Banking \& Finance}, 51, 93--102 (2015)

\bibitem{KM2016} P. Koch-Medina, C. Munari: Unexpected shortfalls of expected shortfall: Extreme default profiles and regulatory arbitrage, {\it Journal of Banking \& Finance}, 62, 141--151 (2016)

\bibitem{kupper2011} M. Kupper, G. Svindland: Dual representation of monotone convex functions on {$L^0$}, {\it Proceedings of the American Mathematical Society}, 139(11), 4073--4086 (2011)

\bibitem{Staum2013} J. Staum: Excess invariance and shortfall risk measures, {\it Operations Research Letters}, 41(1), 47--53 (2013)

\bibitem{yan1980} J.-A. Yan: Caracterisation d'une classe d'ensembles convexes de {$L^1$} ou {$H^1$}, In: {\it S\'eminaire de Probabilit\'es XIV}, 220--222, Springer (1980)

\end{thebibliography}
\end{document}